\documentclass[a4paper,UKenglish]{article}
\pdfoutput=1
\usepackage{hyperref}
\usepackage[protrusion=true,expansion=true]{microtype}
\usepackage[absolute]{textpos}
\usepackage[utf8]{inputenc}
\usepackage{amsfonts, amsopn, amsmath, amsthm}
\usepackage{graphicx}
\usepackage{a4wide}
\usepackage{authblk}   %
\usepackage{enumerate}
\usepackage{macros}

\usepackage{xcolor}
\definecolor{darkblue}{rgb}{0,0,0.45}
\definecolor{darkred}{rgb}{0.6,0,0}
\definecolor{darkgreen}{rgb}{0.13,0.5,0}
\hypersetup{colorlinks, linkcolor=darkblue, citecolor=darkgreen, urlcolor=darkblue}

\title{Diversity of Solutions: An Exploration Through the Lens of
Fixed-Parameter Tractability Theory\footnote{%
An extended abstract of this manuscript has appeared in the Proceedings of the Twenty-Ninth International Joint Conference on Artificial Intelligence, IJCAI 2020~\cite{BasteEtAl20-ijcai2020}.}}

\author[1]{Julien Baste}
\author[2]{Michael R.~Fellows}
\author[2]{Lars Jaffke}
\author[3,4]{Tom\'{a}\v{s} Masa\v{r}\'{i}k}
\author[2]{\\ Mateus de Oliveira Oliveira}
\author[5]{Geevarghese Philip}
\author[2]{Frances A.~Rosamond}

\affil[1]{Ulm University, Germany}
\affil[ ]{\texttt{julien.baste@univ-lille.fr}}
\affil[2]{University of Bergen, Norway} 
\affil[ ]{\texttt{\{michael.fellows,lars.jaffke,mateus.oliveira,frances.rosamond\}@uib.no}}
\affil[3]{%
University of Warsaw, Poland}
\affil[4]{%
Charles University, Prague, Czech Republic}
\affil[ ]{\texttt{masarik@kam.mff.cuni.cz}}
\affil[5]{Chennai Mathematical Institute and UMI ReLaX, India}
\affil[ ]{\texttt{gphilip@cmi.ac.in}}

\date{}

\begin{document}

\maketitle

\begin{abstract}
  When modeling an application of practical relevance as an instance of a
  combinatorial problem X, we are often interested not merely in finding
  \emph{one} optimal solution for that instance, but in finding a
  \emph{sufficiently diverse} collection of good solutions. In this work
  we initiate a systematic study of {\em diversity} from the point of view of
  fixed-parameter tractability theory. First, we consider an intuitive notion of
  \emph{diversity} of a collection of solutions which suits a large variety of
  combinatorial problems of practical interest. We then present an algorithmic
  framework which --\emph{automatically}-- converts a tree-decomposition-based
  dynamic programming algorithm for a given combinatorial problem X into a
  dynamic programming algorithm for the diverse version of X. Surprisingly, our
  algorithm has a polynomial dependence on the diversity parameter. 
\end{abstract}


\renewcommand\P{\ensuremath\mathsf{P}}
\newcommand\poly{\ensuremath\mathsf{poly}}

\section{Introduction}\label{section:Introduction}

In a typical combinatorial optimization problem, we are given a large space of
potential solutions and an objective function. The task is to find a solution
that maximizes or minimizes the objective function. In many situations of
practical relevance, however, it does not really help to get just \emph{one
  optimal} solution; it would be much better to have a small, but
\emph{sufficiently diverse} collection of \emph{sufficiently good} solutions.
Given such a small list of good solutions, we can select one which is best for
our purpose, perhaps by taking into account external factors---such as
aesthetical, political, or environmental---which are difficult or even impossible to
formalize. An early, illustrative example is the problem of generating floor
plans for evaluation by an architect~\cite{galle1989branch}. 

Solution diversity is already a fundamental
concept in many computational tasks. Take, for instance, a web search.
Here, we do not want to find \emph{the one} website that 
`optimally fits' the search term, neither a ranking of a small number of 
`best fits', but what is desirable is a \emph{diverse set} of websites that 
fit the search term \emph{reasonably well}.

Another advantage of considering a set of diverse solutions is that some 
of these solutions may find some use in contexts which are not specified 
\`{a} priori. For instance, in cutting problems~\cite{haessler1991cutting}, which 
are widely studied in the field of operations research, we are given a piece of 
material of standard size and a prescribed set of shapes.
The goal is to cut the material into pieces of the specified shapes in such a way that 
the amount of leftover material is minimized. In this setting, a minimum-size leftover may be 
viewed as a solution. A set of sufficiently diverse solutions
would give the user the opportunity to choose a suitable leftover that could 
be used later in the fabrication of pieces whose shapes have not
been specified in the input of the program.

The notion of diversity has also been applied to solution sets of various types of combinatorial problems.
For instance, the works~\cite{GloverEtAl2000} and~\cite{HebrardHOW05,HebrardEtAl2007}
seek to find solution sets to mixed-integer programming problems 
and constraint satisfaction problems, respectively, that are diverse. 
In other words, the solutions are far apart from each other 
in some mathematical notion of distance.
We refer to~\cite{PetitTrapp2019} for a timely overview of the subject.

From a complexity-theoretic perspective, there are two immediate barriers to this approach.
The first is that most combinatorial problems are already NP-hard when asking only for a single solution.
The second is that the very basic \textsc{Maximum Diversity} problem, which given a set of $n$ elements 
in a metric space and an integer $k < n$,
asks for a size-$k$ subset of the elements such that the sum of the pairwise distances is maximized, 
is NP-hard as well~\cite{KuoGloverDhir1993}.
The theory of \emph{fixed-parameter tractability}~\cite{DowneyFellows13}
provides a powerful framework to overcome these barriers.
The key goal is to identify a secondary numerical measure
of the inputs to an (NP-hard) computational problem,
called the \emph{parameter},
and to provide algorithms in whose runtime
the combinatorial explosion is restricted to the parameter $k$.
More formally, a problem is \emph{fixed-parameter tractable (FPT)},
if it can be solved in time $f(k) \cdot n^c$, where $f$ is a computable function,
$n$ the input size, and $c$ a fixed constant.
On instances where the parameter value is relatively small, 
FPT-algorithms are efficient.
In an application context, we are naturally concerned with finding 
\emph{small} diverse sets of solutions since the aim is to provide the
user with a few alternatives that can then be compared manually.
Therefore, the number of requested solutions is an ideal candidate for parameterization.

In this work, we propose to study the notion of solution diversity
from the perspective of fixed-parameter tractability theory.
We demonstrate the theoretical feasibility of this paradigm 
by showing that diverse variants of 
a large class of parameterized problems admits FPT-algorithms.
Specifically, we consider \emph{vertex-problems} on graphs,
which are sets of pairs $(G, S)$ of a graph $G$ and a subset $S$ of its vertices 
that satisfies some property.
For instance, in the \textsc{Vertex Cover} problem,
we require the set $S$ to be a vertex cover of $G$ 
(i.e.,~$S$ has to contain at least one endpoint of each edge of $G$).
One consequence of our main result which we discuss below in more detail
is that the diverse variant of \textsc{Vertex Cover}, asking for $r$ solutions,
is FPT when parameterized by solution size plus $r$.

Before we proceed, we would like to point out promising future applications
of the \emph{Diverse FPT} paradigm in AI.
The \textsc{Vertex Cover} problem itself naturally models
conflict-resolution: the entities are the vertices of the graph,
and a conflict is represented by an edge.
Now, a vertex cover of the resulting graph is a set of entities
whose removal makes the model conflict-free.
An example of a potential use of \textsc{Diverse Vertex Cover}
in a planning scenario is given in~\cite{BasteEtAl19}.
In general, in planning and scheduling problems,
a large amount of \emph{side information} is lost or intentionally omitted 
in the modeling process.
Some side information could make the model too complex to be solved,
and other information may even be impossible to model.
Offering the user a small number of good solutions to a more easily computable `base model', 
among which they can handpick their favorite solution is a feasible alternative.

\paragraph{A Formal Notion of Diversity} 
We choose a very natural and general measure as our notion of diversity among solutions.
Given two subsets $S$ and $S'$ of a set $V$ the \emph{Hamming
  distance} between $S$ and $S'$ is the number
  \[\hammingdistance(S,S') = |S\backslash S'| + |S'\backslash S|.\]
We define the \emph{diversity of a list} $S_1,\ldots,S_r$ of subsets
of $V$ to be
\[\diversity(S_1,\ldots,S_r) = \sum_{1\leq i<j\leq r} \hammingdistance(S_i,S_j).\]
We can now define the diverse version of vertex-problems:
\begin{definition}[Diverse Problem]\label{definition:DiverseProblem}
Let $\vertexproblem_1,\ldots,\vertexproblem_r$ be vertex-problems,
and let $d\in \N$. We let 
\begin{align*}
\diverse^d(\vertexproblem_1,\ldots,\vertexproblem_r)  =  \{&(G,X_1,\ldots,X_r) \mid
            (G,X_i)\in \vertexproblem_i,\\ 
			&\diversity(X_1,\ldots,X_r)\geq d\}.	
\end{align*}
\end{definition}
\noindent
Intuitively, given vertex-problems $\vertexproblem_1,\ldots,\vertexproblem_r$ and a graph $G$, we want 
to find subsets $S_1,\ldots,S_r$ of vertices of $G$ such that for each $i\in \{1,\ldots,r\}$, $S_i$ 
is a solution for problem $\vertexproblem_i$ on input $G$, and such that the list $S_1,\ldots,S_r$ has 
diversity at least $d$. If all vertex-problems $\vertexproblem_1,\ldots,\vertexproblem_r$ are the same
problem $\vertexproblem$, then we write $\diverse^d_r(\vertexproblem)$ as a shortcut to 
$\diverse^d(\vertexproblem_1,\ldots,\vertexproblem_r)$. 

\paragraph{Diversity and Dynamic Programming}
The treewidth of a graph is a structural parameter that quantifies
how close the graph is to being a forest (i.e., a graph without
cycles).  The popularity of this parameter stems from the fact
that many problems that are NP-complete on general graphs can be
solved in polynomial time on graphs of constant treewidth. In
particular, a celebrated theorem due to
Courcelle~ (see \cite{courcelle1990monadic}) states that any problem
expressible in the monadic second-order logic of graphs can be
solved in polynomial time on graphs of constant treewidth. Besides
this metatheorem, the notion of treewidth has found applications
in several branches of Artificial Intelligence such as Answer Set
Programs~\cite{bliem2017impact}, checking the consistency of
certain relational algebras in Qualitative Spacial
Reasoning~\cite{bodirsky2011rcc8}, compiling Bayesian
networks~\cite{chavira2007compiling}, determining the winners of
multi-winner voting systems~\cite{yang2018multiwinner}, analyzing
the dynamics of stochastic social
networks~\cite{barrett2007computational}, and solving constraint
satisfaction problems~\cite{jegou2007dynamic}. 
A large number of these algorithms are in fact FPT-algorithms when treewidth is the parameter.
Typically, such algorithms are dynamic programming
algorithms which operate on a tree-decomposition in a bottom-up
fashion by computing data from the leaves to the root.

\paragraph{Dynamic Programming Core Model}
We introduce a formalism for dynamic programming based on a tree decomposition, 
which we call the \emph{Dynamic Programming Core} model.
This notion captures a large variety of dynamic programming algorithms on tree decompositions.
We use the model to derive our main result
(Theorem~\ref{theorem:MainTheoremDynamicProgramming}) which is a
framework to efficiently---and automatically---transform
treewidth-based dynamic programming algorithms for vertex-problems
into algorithms for the diverse versions of these problems. More
precisely, we show that if
$\vertexproblem_1,\ldots,\vertexproblem_r$ are vertex-problems
where, for each $i\in \{1,\ldots,r\}$, $\vertexproblem_i$ can be
solved in time $f_i(t)\cdot n^{\Ocal(1)}$, then
$\diverse^d(\vertexproblem_1,\ldots,\vertexproblem_r)$ can be solved
in time $\left(\prod_{i=1}^r f_i(t)\right)\cdot n^{\Ocal(1)}$. 
In particular,
if a vertex-problem $\vertexproblem$ can be solved in
time $f(t)\cdot n^{\Ocal(1)}$, then its diverse version
$\diverse_r^d(\vertexproblem)$ can be solved in time
${f(t)}^{r}\cdot n^{\Ocal(1)}$. 
The surprising aspect of this result is
that the running time depends only \emph{polynomially} on $d$
(which is at most $r^2n$), while a na\"{i}ve dynamic programming
algorithm would have a runtime of $d^{\Ocal(r^2)}\cdot f(t)^r\cdot n^{\Ocal(1)}$.

\paragraph{Discussion of the Diversity Measure}
Various measures of diversity have been used, studied, and compared in
different areas of computer science.
We choose the \emph{sum} of the Hamming distances over all pairs of elements for
this work.
This measure is commonly used for population diversity 
in genetic algorithms~\cite{GaborBPS18,WinebergO03}.
Nonetheless, we would like to point out that it has some weaknesses.
For instance, taking many copies of two disjoint solutions yields a relatively high diversity value, 
and such a solution set is not `diverse' from an intuitive point of view.
We refer to~\cite{BasteEtAl19} for a more detailed discussion.
Another natural measure using the Hamming distance is
the \emph{minimum} Hamming distance over all pairs in a set, as it is done 
e.g.\ in~\cite{HebrardHOW05,HebrardEtAl2007}.
We would like to point out that a straightforward adaptation of our
algorithmic framework would result in a running time of
$d^{\Ocal(r^2)}\cdot f(t)^r\cdot n^{\Ocal(1)}$,
where $d$ is the diversity, $r$ the number of solutions, and $t$ the treewidth.
This remains FPT only when the diversity $d$ 
is an additional component of the parameter,
or when $d$ is naturally upper bounded by $t$ and $r$.
Consider for instance \textsc{Diverse Vertex Cover},
asking for vertex covers of size at most $k$.
In any nontrivial instance, $t$ is at most $k$,
and the Hamming distance between two solutions is at most $2k$, 
therefore we may assume that $d \le 2k$.
This implies that \textsc{Diverse Vertex Cover} can be solved
in time $2^{\Ocal(r^2\log k) + kr}\cdot n^{\Ocal(1)}$
using the minimum Hamming distance as a diversity measure.

\paragraph{Related Work}
The above-mentioned \textsc{Maximum Diversity} problem
has applications in the generation of diverse query results, 
see e.g.~\cite{GollapudiSharma2009,AbbassiMirrokniThakur2013}.
Besides mixed integer programming~\cite{GloverEtAl2000,DannaWoodruff2009,PetitT15},
binary integer linear programming~\cite{GreistorferEtAl2008,TrappKonrad2015}
and constraint programming~\cite{HebrardHOW05,HebrardEtAl2007},
diverse solution sets have been considered in SAT solving~\cite{Nadel2011},
recommender systems~\cite{AdomaviciusKwon2014},
routing problems~\cite{SchittekatS2009},
answer set programming~\cite{EiterEtAl2013}, and
decision support systems~\cite{LokketangenW2005,HadzicEtAl2009}.

\section{Preliminaries}\label{section:Preliminaries}

For positive integers $a$ and $b$, with $a < b$, %
we use \(\intv{a,b}\) to
denote the set \(\{a,a+1,\ldots,b\}\). We use \(V(G)\) and \(E(G)\),
respectively, to denote the vertex and edge sets of a graph
\(G\). For a tree $T$ rooted at $\roots$ we use $T_t$ to
denote the subtree of $T$ rooted at a vertex \(t\in V(T)\). 
A \emph{rooted tree decomposition} of a graph $G$ 
is a tuple ${\cal D}=(T,\roots,{\cal X})$, where $T$ is a tree rooted at $\roots \in V(T)$
and ${\cal X}=\{X_{t}\mid t\in V(T)\}$ is a collection of subsets of $V(G)$
such that:
\begin{itemize}
\item $\bigcup_{t \in V(T)} X_t = V(G)$,
\item for every edge $\{u,v\} \in E(G)$, there is a $t \in V(T)$ such that $\{u, v\} \subseteq X_t$, and
\item for each $\{x,y,z\} \subseteq V(T)$ such that $z$ lies on the unique path between $x$ and $y$ in $T$,  $X_x \cap X_y \subseteq X_z$.
\end{itemize}
We say that the vertices of $T$ are the \emph{nodes} of ${\cal D}$ and that the sets in ${\cal X}$ are the \emph{bags} of ${\cal D}$.
Given a node $t \in V(T)$, we denote by $G_t$ the subgraph of $G$
induced by the set of vertices  
\(\bigcup_{s\in V(T_t)} X_s.\) 
The \emph{width} of a  tree decomposition 
${\cal D}=(T,q, {\cal X})$ is defined as $\max_{t \in V(T)} |X_t| - 1$.
The \emph{treewidth} of a graph $G$, denoted by $\tw(G)$, is the smallest integer $w$ such that there exists a rooted tree decomposition of $G$ of width at most $w$.
The \emph{rooted path decomposition} of a graph is a rooted tree decomposition ${\cal D}=(T,\roots,{\cal X})$ such that $T$ is a path and $\roots$ is a vertex of degree $1$.
The \emph{pathwidth} of a graph $G$, denoted by $\pw(G)$, is the smallest integer $w$ such that there exists a rooted path decomposition of $G$ of width at most $w$.
Note that in a rooted path decomposition, every node has at most one child. 

For convenience we will always assume that the bag associated to
the root of a rooted tree decomposition is empty. For a node
\(t\in{}V(T)\) we use $\delta_{\mathcal{D}}(t)$, or $\delta(t)$ when $\mathcal{D}$ is clear from the context, to denote the
number of children of $t$ in the tree $T$.  For nodes $t$ and $t'$
of $V(T)$ where $t'$ is the parent of $t$ we use
$\forgotten(t) = {X}_{t} \setminus {X}_{t'}$ to
denote the set of vertices of $G$ which are \emph{forgotten} at
$t$.  By convention, for the root $\roots$ of $T$, we let
$\forgotten(\roots) = \emptyset$.  For \(t\in{}V(T)\) we
denote by $\new(t)$ the set
$X_t \setminus \bigcup_{i=1}^{\delta(t)} X_{t_i}$ where
$t_1, \ldots, t_{{\delta(t)}}$ are the children of $t$.
Given a rooted tree decomposition \(\mathcal{D}\) of a graph \(G\)
one can obtain, in linear time, a tree decomposition
\((T,q, \mathcal{X})\) of \(G\) of the same width as
\(\mathcal{D}\) such that for each $t \in V(T)$,
$\delta(t) \leq 2$ and $|\new(t)| \leq 1$~\cite{Cygan}.  From now
on we assume that every rooted tree decomposition is of this kind.

\section{A First Example: Diverse Vertex Cover}\label{section:VertexCover}

The main result of this paper is a general framework to
automatically translate tree-decomposition-based dynamic
programming algorithms for ver\-tex-problems into algorithms for the
diverse versions of these problems. We develop this framework in
Section~\ref{section:DynamicProgramming}.  In this section we
illustrate the main techniques used in this conversion process by
showing how to translate a tree-decomposition-based dynamic
programming algorithm for the \textsc{Vertex Cover} problem into
an algorithm for its diverse version \textsc{Diverse Vertex
  Cover}.
Given a graph $G$ and three integers $k$, $r$, and $d$, the
\textsc{Diverse Vertex Cover} problem asks whether one can find
$r$ vertex covers in $G$, each of size at most $k$, such that
their diversity is at least $d$.
Our algorithm for this problem will run in $2^{\Ocal{(kr)}} |V(G)|$ time.

\newcommand{\indicatorfunction}{\gamma}

\subsection{Incremental Computation of Diversity}
Recall that we defined the diversity of a list
$S_1,S_2,\ldots,S_r$ of subsets of a set $V$ to be
\[ \diversity(S_1,\ldots,S_r) = \sum_{1\leq i<j\leq r} \hammingdistance(S_i,S_j).\]
We will now describe a way to compute the diversity
\(\diversity(S_1,\ldots,S_r)\) in an incremental fashion, by incorporating the influence of each
element of $V$ in turn. For each element $v\in V$ and each pair of
subsets $S,S'$ of $V$, we define $\indicatorfunction(S,S',v)$ to
be $1$ if \(v\in(S\setminus{}S')\cup(S'\setminus{}S)\), and to be
$0$ otherwise. Intuitively, $\indicatorfunction(S,S',v)$ is $1$ if
and only if the element $v$ contributes to the Hamming distance
between $S$ and $S'$. Given this definition we can rewrite
$\hammingdistance(S,S')$ as

\[\hammingdistance(S,S') = \sum_{v\in V} \indicatorfunction(S,S',v),\]

and the diversity of a list $S_1,\ldots,S_r$ of subsets of $V$ as 

\[\begin{array}{ll}
	\diversity(S_1,\ldots,S_r)    & =  \sum_{1\leq i<j\leq r}  \sum_{v\in V} \indicatorfunction(S_i,S_j,v) \\

    \\
				   & = \sum_{v\in V} |\{\ell\;:\; v\in S_{\ell}\}|\cdot |\{\ell\;:\; v\notin S_{\ell}\}|. 
\end{array}
\]

Now, if we define the \emph{influence} of $v$ on the list  $S_1,\ldots,S_r$ as \[\diverseinfluence(S_1,\ldots,S_r,v) =
|\{\ell\;:\; v \in S_{\ell}\}| \cdot |\{\ell\;:\; v \not\in
S_{\ell}\}|,\] then we have that

\begin{equation}
\label{equation:Incremental}	
\diversity(S_1,\ldots,S_r) = \sum_{v\in V} \diverseinfluence(S_1,\ldots,S_r,v).
\end{equation}

\subsection{From Vertex Cover to Diverse Vertex Cover}
We now solve \textsc{Diverse Vertex Cover} using dynamic
programming over a tree decomposition of the input graph.  An
excellent exposition of tree-width-based dynamic programming
algorithms can be found in~\cite[Chapter~7]{Cygan}.

Let $(G,k,r,d)$ be an instance of {\sc Diverse Vertex Cover} and let
$\mathcal{D} = (T,\roots,\mathcal{X})$ be a rooted tree decomposition of $G$. 
For each node $t \in V(T)$, we define the set 
\[
\Ical_t  =  \{((S_1,s_1),\ldots,(S_r,s_r),\ell) \mid \ell \in \intv{0,d}, \forall i \in \intv{1,r}, S_i \subseteq X_t, s_i \in \intv{0,k}\}.
\]

This set $\Ical_t$, $t \in V(T)$, is such that the partial solutions we will construct for the node $t$ will always be a subset of $\Ical_t$.
Note that for each $t\in V(T)$, $|\Ical_t| \leq (2^{|X_t|} \cdot (k+1))^r\cdot (d+1)$. Now,
our dynamic programming algorithm for {\sc Diverse Vertex Cover} consists in constructing for 
each $t\in V(T)$ a subset $\Rcal_t \subseteq \Ical_t$ as follows. 
Let $t$ be a node in $V(T)$ with children $t_1,\ldots, t_{\delta(t)}$.
We recall that, by convention, this set of children is of size $0$, $1$, or $2$. We let 
$\Rcal_t$ be the set of all tuples $((S_1,s_1),\ldots,(S_r,s_r),\ell) \in \Ical_t$ satisfying 
the following additional properties: 
\begin{enumerate}
\item For each $j \in \intv{1,r}$, $E(G[X_t\setminus S_j]) = \emptyset$. 
\item For each $i \in \intv{1,\delta(t)}$ there exists a tuple 
  $((S_1^i,s_1^i),\ldots,(S_r^i,s_r^i),\ell_i)$ in $\Rcal_{t_i}$ such that 
  \begin{enumerate}
  \item $S_j \cap X_{t_i} = S_j^i \cap X_t$ for each $i\in \intv{1,\delta(t)}$ and each 
    $j \in \intv{1,r}$,
  \item For each $j \in \intv{1,r}$, 
    $s_j = | \forgotten(t) \cap S_j| + \sum_{i=1}^{\delta(t)}s_j^i$,
  \item and  $\ell = \min (d,m)$ where  $m = \sum_{v\in \forgotten(t)} \diverseinfluence (  S_1,\ldots,  S_r , v ) + \sum_{i=1}^{\delta(t)}\ell_i. $
  \end{enumerate}
\end{enumerate}

\begin{lemma}\label{lemma:Correctness}
$(G,k,r,d)$ is a \yes{}-instance of {\sc Diverse Vertex Cover} if and only if there is a tuple
$((S_1,s_1),\ldots,(S_r,s_r),\ell)$ in $\Rcal_q$ such that $\ell = d$.
\end{lemma}

\begin{proof}

Using induction, one can see that 
for each $t\in V(T)$, 
$\Rcal_t$ is the set of every element of $\Ical_t$ such that,
with $Y_t = X_t \setminus \forgotten(t)$,
there exists $(\hs_1, \ldots, \hs_r) \in V(G_t)^r$,
that satisfies:
\begin{itemize}
\item for each $i \in \intv{1,r}$, $\hs_i$ is a vertex cover of $G_t$,
\item for each $i \in \intv{1,r}$, $\hs_i \cap X_t = S_i$,
\item for each $i \in \intv{1,r}$, $|\hs_i \setminus Y_t| = s_i$, and
\item $\min(d,\diversity(\hs_1\setminus Y_t, \ldots, \hs_r\setminus Y_t)) = \ell$.
\end{itemize}

As the root $q$ of the tree decomposition $\mathcal{D}$ is such that $X_q = \emptyset$, we obtain that the elements in $\Rcal_q$ are the elements $((\emptyset,s_1),\ldots,(\emptyset,s_r),\ell)$ of $\Ical_q$ such that there exists $(\hs_1, \ldots, \hs_r) \in V(G)^r$, that satisfy, 
\begin{itemize}
\item for each $i \in \intv{1,r}$, $\hs_i$ is a vertex cover of $G_t$,
\item for each $i \in \intv{1,r}$, $|\hs_i| = s_i \leq k$, and
\item $\min(d,\diversity(\hs_1, \ldots, \hs_r)) = \ell$.
\end{itemize}
As such, a tuple $(\hs_1, \ldots, \hs_r)$ of subsets of $V(G)$ is a solution of 
{\sc Diverse Vertex Cover} if and only if $\ell \geq d$, the lemma follows.
\end{proof}

\begin{theorem}\label{theorem:DiverseVC}
Given a graph $G$, integers $k,r,d$, and a rooted tree decomposition $\mathcal{D} = (T,\roots,\mathcal{X})$ of $G$ of width $w$,
one can determine whether $(G,k,r,d)$ is a \yes{}-instance of {\sc Diverse Vertex Cover} in time 
\[\Ocal(2^{r}\cdot (2^{w+1} \cdot (k+1))^{a\cdot r}\cdot d^a \cdot w \cdot r  \cdot n),\]
where  $a = \max_{t \in V(T)}\delta(t) \leq 2$ and $n = |V(T)|$.
\end{theorem}
\begin{proof}
Let us analyze the time needed to compute $\Rcal_q$.
We have that, for each $t \in V(\mathcal{D})$, $|\Ical_t| \leq (2^{\cdot|X_t|} \cdot (k+1))^r\cdot (d+1)$.
Note that
given $I_1, \ldots, I_{\delta(t)}$ be elements of $\Rcal_{t_1}, \ldots, {\Rcal_{t_{\delta(t)}}}$, there are at most $2^{|\new(t)|\cdot r} \leq 2^r$ ways to create an element $I$ of $\Rcal_t$ by selecting, or not the (potential) new element of $X_t$ for each set $S_i$, $i \in \intv{1,r}$.
The remaining is indeed fixed by $I_1, \ldots, I_{\delta(t)}$.
Thus, $\Rcal_t$ can be computed in time $\Ocal(r \cdot |X_t| \cdot 2^{r}\cdot \prod_{i=1}^{\delta(t)}|\Rcal_{t_i}|)$, where the factor $r \cdot |X_t|$ appears when verifying that the element we construct satisfy
$\forall j \in \intv{1,r}, E(G[X_j\setminus S_j]) = \emptyset$.
As we need to compute $\Rcal_t$ for each $t \in V(\mathcal{D})$ and that $|V(\mathcal{D})|  = \Ocal(n)$ and we can assume that $\delta(t) \leq 2$ for each $t \in V(\mathcal{D})$, the theorem follows.
\end{proof}
\begin{remark}\label{rem:pw}
Given a graph $G$ and a vertex cover $Z$ of $G$ of size $k$, one
can find a rooted path decomposition
$\mathcal{D} = (T,q,\mathcal{X})$ of $G$ of width $k$, in linear time. 
\end{remark}
This can be done by considering the bags
$Z \cup \{v\}$ for each $v \in V(G)$ in any fixed order.
Thus, from Theorem~\ref{theorem:DiverseVC}, we get the following corollary, which establishes an upper bound 
for the running time of our dynamic programming algorithm for {\sc Diverse Vertex Cover} solely in terms of the size $k$ of the vertex cover, the number $r$ of requested solutions, and the diversity $d$. 

\begin{corollary}\label{corollary:DVC}
  {\sc Diverse Vertex Cover} can be solved on an input $(G,k,r,d)$ in time 
  \(\Ocal((2^{k+2} \cdot (k+1))^{r} \cdot d \cdot k \cdot r \cdot |V(G)|).\)
\end{corollary}

\section{Computing Diverse Solutions using the Dynamic Programming Core model}\label{section:DynamicProgramming}

In this section, we show that the process illustrated in Section \ref{section:VertexCover},
of lifting a dynamic programming  algorithm for a combinatorial problem to an algorithm for its diverse version,
can be generalized to a much broader context. As a first step, we introduce the notion of 
{\em dynamic programming core}, a suitable formalization 
of the {\em intuitive} notion of tree-width based dynamic programming that satisfies three essential properties.
First, this formalization is general enough to be applicable to a large class of combinatorial optimization 
problems. Second, this formalization is compatible with the notion of diversity, in the sense that the
lifting of an algorithm for a problem to an algorithm for the diverse version of this problem can be 
done automatically, without requiring human ingenuity. Third, the resulting lifted algorithm is fast when 
compared with the original one. In particular, the running time of the resulting algorithm is polynomial 
on the diversity parameter. This is a highly desired property since this allows our framework 
to be applied in the context where the sizes of the considered solution sets are not bounded. 

Below, we let $\graphs$ be the set of simple, undirected graphs whose vertex
set is a finite subset of $\N$. We say that a subset
$\graphproblem \subseteq \graphs$ is a graph problem. Intuitively, 
a dynamic programming algorithm working on tree decompositions 
may be understood as a procedure that takes a graph $G\in \graphs$ 
and a rooted tree decomposition $\mathcal{D}$
of $G$ as input, and constructs a certain amount of data for each node of
$\mathcal{D}$.  The data at node $t$ is constructed by induction
on the height of $t$, and in general, this data is used to encode
the existence of a partial solution on the graph induced by bags
in the sub-tree of $\mathcal{D}$ rooted at $t$. 
In the below definition, this is captured in the relation $\process_{\acore,\agraph,\mathcal{D}}(t)$.
Such an algorithm
accepts the input graph $G$ if the data associated with the root
node contains a string belonging to a set of \emph{accepting strings},
captured below in the set $\accepting_{\acore,\agraph,\mathcal{D}}$.
We formalize this intuitive notion in the following concept of 
\emph{dynamic programming core}.

\begin{definition}[Dynamic Programming Core]\label{definition:DynamicCore}
	A \emph{dynamic programming core} is an algorithm $\acore$ that takes a graph $\agraph\in \graphs$ and 
	a rooted tree decomposition $\mathcal{D}$ of $\agraph$ as input, and produces the following data. 
\begin{itemize}
	\item A finite set $\accepting_{\acore,\agraph,\mathcal{D}} \subseteq 2^{\{0,1\}^*}$.
        \item A finite set $\process_{\acore,\agraph,\mathcal{D}} (t) \subseteq \left(2^{\{0,1\}^*}\right)^{\delta(t)+1}$ for each $t \in V(\mathcal{D})$.
\end{itemize}
\end{definition}

We let $\timecomplexity(\acore,G,\mathcal{D})$ be the overall time necessary to construct 
the data associated with all nodes of $\mathcal{D}$. The size of $\acore$ on a pair 
$(\agraph,\mathcal{D})$ is defined as 
\[\size(\acore,\agraph,\mathcal{D}) = \max  \{|\process_{\acore,\agraph,\mathcal{D}}(t)| \mid t\in V(\mathcal{D})\}.\]

Next, we define the notion of a {\em witness} for a dynamic programming core. Intuitively such witnesses are certificates of the existence of a solution. 

\begin{definition}\label{definition:Witness}
Let $\acore$ be a dynamic programming core, $\agraph$ be a graph in $\graphs$, and  
$\mathcal{D} = (T,q, \mathcal{X})$ be a rooted tree decomposition of 
$\agraph$. A \emph{$(\acore,\agraph,\mathcal{D})$-witness} is a function $\alpha:V(T)\rightarrow \{0,1\}^*$
such that the following conditions are satisfied. %
\begin{enumerate}
\item For each $t \in V(T)$, with children $t_1, \ldots, t_{\delta(t)}$,
  $(\alpha(t),\alpha(t_1), \ldots, \alpha(t_{\delta(t)})) \in \process_{\acore,\agraph,\mathcal{D}}(t)$.
\item  $\alpha(q)\in \accepting_{\acore}$.
\end{enumerate}
\end{definition}

Using the notion of witness, we define formally what it means for a dynamic programming core to solve a 
combinatorial problem. 

\begin{definition}\label{definition:Solvability}
  We say that a dynamic programming core $\acore$ solves a problem  $\mathcal{P}$ if for each graph 
$\agraph\in \graphs$, and each rooted tree decomposition $\mathcal{D}$ of $\agraph$,
$\agraph\in \graphproblem$ if and only if a $(\acore,\agraph,\mathcal{D})$-witness exists. 
\end{definition}

\begin{theorem}\label{theorem:DynamicSolvability}
Let $\graphproblem$ be a graph problem and $\acore$ be a dynamic programming core that solves $\graphproblem$. 
Given a graph $\agraph\in \graphs$ and a rooted tree decomposition $\mathcal{D}$ of 
$\agraph$, one can determine whether $\agraph\in \graphproblem$ in time
$$\mathcal{O}\left(\sum_{t \in V(T)}|\process_{\acore,\agraph,\mathcal{D}}(t)| + \timecomplexity(\acore,G,\mathcal{D})\right).$$
\end{theorem}
\begin{proof}
    Given $\acore$, $\agraph$, and $\mathcal{D} = (T, q, \mathcal{X})$, we construct the set $\accepting_{\acore,\agraph,\mathcal{D}}$ and 
	the sets $\process_{\acore,\agraph,\mathcal{D}}(t)$ for each $t \in V(\mathcal{D})$.
    By definition, this can be done in time $\timecomplexity(\acore,G,\mathcal{D})$.

    Given $t \in V(T)$ and $w \in \{0,1\}^*$, a \emph{$(\acore,G,\mathcal{D},t,w)$-witness} is
    a function \[\beta: V(\subtree(T,t)) \to \{0,1\}^*\] such that 
    for each $t' \in V(\subtree(T,t))$, with children $t_1, \ldots, t_{\delta(t)}$,
    $$(\beta(t'),\beta(t_1), \ldots, \beta(t_{\delta(t)})) \in \process_{\acore,\agraph,\mathcal{D}}(t)$$ and
    $\beta(t) = w$.
    Note that there exists a $(\acore,\agraph,\mathcal{D})$-witness if and only if there exists a $(\acore,G,\mathcal{D},q,w)$-witness for some $w \in \{0,1\}^*$.

    For each $t \in V(T)$, we define $\potentialSolutions(G,\mathcal{D},t)$ to be the set of every $w \in \{0,1\}^*$ such that there exists a $(\acore,G,\mathcal{D},t,w)$-witness.
    Let $t \in V(T)$ and assume that we are able to construct $\potentialSolutions(G,\mathcal{D},t_i)$ for every $i \in [1,\delta(t)]$ where $t_1, \ldots, t_{\delta(t)}$ are the children of $t$.
    We can then construct $\potentialSolutions(G,\mathcal{D},t)$ as follows.
    For each $(w,w_1, \ldots, w_{\delta(t)}) \in \process_{\acore,\agraph,\mathcal{D}}(t)$, we add $w$ to $\potentialSolutions(G,\mathcal{D},t)$ if for each $i \in [1,\delta(t)]$, $w_i \in \potentialSolutions(G,\mathcal{D},t_i)$.
    It is easy to see that for each such $w$, there exists a $(\acore,G,\mathcal{D},t,w)$-witness that is an extension of the $(\acore,G,\mathcal{D},t_i,w_i)$-witness, $i \in [1,\delta(t)]$.
    Moreover if there exists a $(\acore,G,\mathcal{D},t,w)$-witness $\beta$ for some $w \in \{0,1\}^*$, then, for each $i \in [1,\delta(t)]$, the restriction of $\beta$ to $\subtree(T,t_i)$ is a $(\acore,G,\mathcal{D},t_i,w_i)$-witness for some $w_i \in \{0,1\}^*$, and so, by induction hypothesis, $w_i \in \potentialSolutions(G,\mathcal{D},t_i)$.
    This implies that our construction has correctly added $w$ to $\potentialSolutions(G,\mathcal{D},t)$.
    Thus $\potentialSolutions(G,\mathcal{D},t)$ is correctly constructed.
    
    From  Definition~\ref{definition:Solvability} we have that  $\agraph\in \graphproblem$ if and only if $\potentialSolutions(G,\mathcal{D},q)\not = \emptyset$.
    Note that the time needed to construct $\potentialSolutions(G,\mathcal{D},q)$ is
    $\mathcal{O}\left(\sum_{t \in V(T)}|\process_{\acore,\agraph,\mathcal{D}}(t)|\right)$.
    Therefore, the theorem follows. 
\end{proof}

\subsection{Dynamic Programming Cores for Vertex Problems}

Let $\acore$ be a dynamic programming core. A \emph{$\acore$-vertex-membership function} is a function
$\membershipfunction:\N\times \{0,1\}^*\rightarrow \{0,1\}$ such that 
for each graph $G$, each rooted tree decomposition 
$\mathcal{D}= (T,q,\mathcal{X})$ of $G$ and each $(\acore,\agraph,\mathcal{D})$-witness $\alpha$, 
it holds that $\membershipfunction(v,\alpha(t)) = \membershipfunction(v,\alpha(t'))$ for each 
edge $(t,t')\in E(T)$ and each vertex $v\in X_t\cap X_{t'}$. Intuitively, if $\agraph$ is a graph and $\mathcal{D}$ is a rooted tree decomposition of $\agraph$, then 
a $\acore$-vertex-membership together with a $(\acore,\agraph,\mathcal{D})$-witness, provide
an encoding of a subset of vertices of the graph. More precisely, we let 
\[S_{\membershipfunction}(\agraph,\mathcal{D},\alpha) = \{v\mid \exists t\in V(T_{\mathcal{D}}), 
\membershipfunction(v,\alpha(t)) = 1\}\]
be this encoded vertex set. 
Given a $\acore$-vertex-membership function $\membershipfunction$, we let 
$\hat{\membershipfunction}: \{0,1\}^*\rightarrow 2^{\N}$ be the function that sets 
$\hat{\membershipfunction}(w) = \{v \in \N\mid \membershipfunction(v,w) = 1\}$ for each 
$w \in \{0,1\}^*$.

Let $\vertexproblem$ be a vertex-problem, $\acore$ be a dynamic programming core, and $\membershipfunction$ be a 
$\acore$-vertex-membership function. We say that \emph{$(\acore,\membershipfunction)$ solves $\vertexproblem$}
if for  each graph $\agraph\in \graphs$, each subset $S\subseteq V(\agraph)$, and 
each rooted tree decomposition $\mathcal{D}$, $(\agraph,S)\in \vertexproblem$
if and only if there exists a $(\acore,\agraph,\mathcal{D})$-witness $\alpha$ such that  
$S = S_{\membershipfunction}(\agraph,\mathcal{D},\alpha)$. 

The following theorem is the main result of this section. It shows how to transform dynamic programming cores 
for problems $\vertexproblem_1,\dots,\vertexproblem_r$ into a dynamic programming core for the problem 
$\diverse^d(\vertexproblem_1,\dots,\vertexproblem_r)$. 

\begin{theorem}\label{theorem:MainTheoremDynamicProgramming}
Let $\vertexproblem_1,\ldots,\vertexproblem_r$ be vertex-problems, let 
$(\acore_i,\membershipfunction_i)$ be a dynamic programming core for $\vertexproblem_i$, and let $d$ be an integer.  
The problem $\diverse^d(\vertexproblem_1,\ldots,\vertexproblem_r)$, on graph $\agraph$ with rooted tree decomposition 
	$\mathcal{D}=(T,\roots,\mathcal{X})$, can be solved in time 
$$\mathcal{O}(d^a\cdot |V(T)|\cdot \prod_{i=1}^{r}\size(\acore_i,\agraph,\mathcal{D}) + 
\sum_{i=1}^{r}\timecomplexity(\acore_i,\agraph,\mathcal{D})),$$ where $a = \max_{t \in V(T)}\delta(t) \leq 2$.
\end{theorem}
\begin{proof} 
Let $w_1,\ldots,w_r \in \{0,1\}^*$ and $v\in V(G)$.
We extend the definition of diverse influence to $w_1,\ldots,w_r$ such that
\[\diverseinfluence(w_1,\ldots,w_r,v) = \diverseinfluence(\hat{\membershipfunction}_1(w_1),\ldots,\hat{\membershipfunction}_r(w_r),v).\]
Before proving Theorem \ref{theorem:MainTheoremDynamicProgramming}, we state and prove the following 
technical lemma.

\begin{lemma}\label{lemma:Intermediate}
Let $G$  be a graph and $\mathcal{D} = (T,q,\mathcal{X})$ be a rooted tree decomposition of $G$. 
 $(G,Z_1,\ldots,Z_r)$ belongs to $\diverse^{d}(\vertexproblem_1,\ldots,\vertexproblem_r)$ if and only if
there exist $\alpha_1,\ldots,\alpha_r: V(T)\rightarrow \{0,1\}^*$ such that the following conditions are satisfied. 
\begin{enumerate}
\item For each $i\in [1,r]$, $\alpha_i$ is a $(\acore_i,G,\mathcal{D})$-witness
          and  $Z_i = S_{\membershipfunction_i}(\agraph,\mathcal{D},\alpha_i)$. 
\item $\sum_{t\in V(\mathcal{D})} \sum_{v \in \forgotten(t)} \diverseinfluence(\alpha_1(t),\ldots,\alpha_r(t),v) \geq d$. 
\end{enumerate}
\end{lemma}
\begin{proof}
  First assume that $(G,Z_1,\ldots,Z_r)$ belongs to $\diverse^{d}(\vertexproblem_1,\ldots,\vertexproblem_r)$.
  By Definition~\ref{definition:DiverseProblem}, for each $i \in [1,r]$, we have that $(G,Z_i) \in \vertexproblem_i$,
  and so, there exists a $(\acore_i,\agraph,\mathcal{D})$-witness $\alpha_i$ such that  $Z_i = S_{\membershipfunction_i}(\agraph,\mathcal{D},\alpha_i)$.
  Thus Condition~1 is satisfied.
  Moreover, we have that for each $t \in V(\mathcal{D})\setminus \{\roots\}$ and each $v \in X_t$,
  $\diverseinfluence(\alpha_1(t),\ldots,\alpha_r(t),v) = \diverseinfluence(Z_1,\ldots,Z_r,v)$.
  Together with the fact that each vertex is in exactly one set $\forgotten(t)$, $t \in V(\mathcal{D})\setminus\{q\}$, and $\diversity(Z_1,\ldots,Z_r)\geq d$ imply Condition~2.

  Assume now that there exist $\alpha_1,\ldots,\alpha_r: V(T)\rightarrow \{0,1\}^*$ that satisfy Conditions 1 and 2.
  Condition~1 implies that for each $i \in [1,r]$, $(G,Z_i) \in \vertexproblem_i$.
	Moreover, as for each $v \in V(G)$, there is exactly one node ${t \in V(\mathcal{D})\setminus\{q\}}$ such that $v \in \forgotten(t)$, by definition of a rooted tree decomposition, 
  Condition~2 implies that $\diversity(Z_1,\ldots,Z_r)\geq d$.
  Thus, $(G,Z_1,\ldots,Z_r)$ belongs to $\diverse^{d}(\vertexproblem_1,\ldots,\vertexproblem_r)$.
\end{proof}

Now we are in a position to prove Theorem \ref{theorem:MainTheoremDynamicProgramming}. 

For each $i\in \{1,\ldots,r\}$, we start by constructing the data corresponding to the dynamic 
core $\acore_i$. The overall construction takes time $\sum_{i=1}^r \timecomplexity(\acore_i,\agraph,\mathcal{D})$. 

Subsequently, we define a dynamic core $\acore$ for the problem $\diverse^d(\vertexproblem_1,\ldots,\vertexproblem_r)$. 
Let $\agraph\in \graphs$ and $\mathcal{D}= (T, q, \mathcal{X})$ be a rooted tree decomposition of $\agraph$. The dynamic core
$\acore$ produces the following data. 
\begin{itemize}
\item $\accepting_{\acore} = \{(w_1,\ldots,w_r,d)\mid \forall  i \in [1,r], w_i \in \accepting_{\acore_i}\}$.
\item For each $t \in V(\mathcal{D})$, $\process_{\acore,\agraph,\mathcal{D}} (t) =$  
$$
\begin{array}{l}
\{((w_1,\ldots,w_r,\ell), (w_1^1,\ldots,w_r^1,\ell^1), 
	\ldots,(w_1^{\delta(t)},\ldots,w_r^{\delta(t)},\ell^{\delta(t)}))\mid \\ 
	\quad \forall i \in [1,r],  (w_i,w_i^1, \ldots, w_i^{\delta(t)}) \in \process_{\acore_i,\agraph,\mathcal{D}} (t),\\ 
 \quad s = \sum_{i \in [1,\delta(t)]} \ell^i + \sum_{v \in \forgotten(t)} I(w_1,\ldots,w_r,v),
 \ell = \min \{s,d\} \}.
\end{array}		
$$ 
\end{itemize}

Let $\alpha$ be a $\acore$-witness of $(G,\mathcal{D})$, let $\alpha_i$ be the projection of $\alpha$
to its $i$-th coordinate, and let $\beta$ be the projection of $\alpha$ to its last coordinate. Then we 
have that $\alpha$ is a $(\acore,\agraph,\mathcal{D})$-witness for $(G,\mathcal{D})$ if and only if $\alpha_i$ is a $(\acore_i,\agraph,\mathcal{D})$-witness
for $(G,\mathcal{D})$, and  for $q$ being the root of $\mathcal{D}$,
\[
 \beta(q) 
 =  \min\{d,\sum_{t\in V(\mathcal{D})}\sum_{v \in \forgotten(t)} \diverseinfluence(\alpha_1(t),\ldots,\alpha_r(t),v)\}\geq d.\]

 By Lemma~\ref{lemma:Intermediate}, we have that this happens if and only if 
 \[(G,S_{\membershipfunction_1}(\agraph,\mathcal{D},\alpha_{1}),\ldots, 
 S_{\membershipfunction_r}(\agraph,\mathcal{D},\alpha_{r}))\] 
 belongs to $\diverse^{d}(\vertexproblem_1,\ldots,\vertexproblem_r)$.

 Let now analyze the running time of this procedure.
 When constructing $\process_{\acore,\agraph,\mathcal{D}} (t)$ for some $t \in V(T)$, we need to combine every combination of elements of $\process_{\acore_i,\agraph,\mathcal{D}} (t)$, $i \in \intv{1,r}$ and of values of $\ell^i$, $i \in \intv{1,{\delta(t)}}$.
 This can be done in time $\Ocal(d^{\delta(t)}\cdot |V(T)|\cdot \prod_{i=1}^{r}\size(\acore_i,\agraph,\mathcal{D}))$.
 Thus constructing the data associated to $\acore$, $\agraph$, and $\mathcal{D}$ takes
 \[\Ocal(d^{\delta(t)} \cdot |V(T)|\cdot \prod_{i=1}^{r}\size(\acore_i,\agraph,\mathcal{D})+ \sum_{i=1}^{r}\timecomplexity(\acore_i,\agraph,\mathcal{D})).\]

 Moreover, as for every $t \in V(T)$, $|\process_{\acore,\agraph,\mathcal{D}} (t)| \leq d^{\delta(t)}\cdot \prod_{i=1}^{r}\size(\acore_i,\agraph,\mathcal{D})$, then by Theorem~\ref{theorem:DynamicSolvability},
 $\diverse^d(\vertexproblem_1,\ldots,\vertexproblem_r)$ can be solved in time
 $\Ocal(d^{a}\cdot |V(T)| \cdot \prod_{i=1}^{r}\size(\acore_i,\agraph,\mathcal{D}))$ where $a = \max_{t \in V(T)}\delta(t) \leq 2$.
 The theorem follows. 
\end{proof}

\subsection{An Illustrative Application of Theorem \ref{theorem:MainTheoremDynamicProgramming}}

In this subsection we show how to apply Theorem~\ref{theorem:MainTheoremDynamicProgramming} in 
the construction of an improved dynamic programming algorithm for {\sc Diverse Vertex Cover}.
The first thing to do is to describe a dynamic programming core  $\acore_{\rm VC}$ for {\sc $k$-Vertex Cover}. 
Given a graph $G$ and a rooted tree decomposition $\mathcal{D} = (T,\roots,\mathcal{X})$, this dynamic programming core $\acore_{\rm VC}$ produces:
\begin{align*}
	\accepting_{\acore,\agraph,\mathcal{D}} &= \{(S,s) \mid S \subseteq X_\roots, s \leq k\} \\
	\process_{\acore,\agraph,\mathcal{D}}(t) &=
	  \{((S,s),(S^1,s^1), \ldots, (S^{\delta(t)}, s^{\delta(t)})) \mid \\ 
	 	& E(G[X_t \setminus S]) = \emptyset, \\
	 	&\forall i \in \intv{1,{\delta(t)}}\colon S^i \cap X_t = S \cap X_{t_i}, \\
		&s = |\forgotten(t)\cap S| + \sum\nolimits_{t=1}^{\delta(t)}s^i\}
\end{align*}

Provided the width of the decomposition is at most $k$, this can be done in time $\Ocal((2^{k+1} \cdot (k+1))^{\delta(t)}\cdot k\cdot {\delta(t)})$ for each $t \in V(T)$, where the factor $k\cdot {\delta(t)}$ appears as we need the conditions $E(G[X_t \setminus S]) = \emptyset$ and $\forall i \in \intv{1,{\delta(t)}}, S^i \cap X_t = S \cap X_{t_i}$ to be verified.
It is easy to verify that $\acore_{\rm VC}$ is a dynamic programming core for the  {\sc Vertex Cover} problem.
As described in Remark~\ref{rem:pw}, we know that
we can construct a rooted path decomposition of $G$ of width $k$.
We are now considering this rooted path decomposition.
Thus, for each $t \in V(T)$, $|\process_{\acore,\agraph,\mathcal{D}} (t)| \leq 2 \cdot 2^{k+1} \cdot (k+1)$.
By Theorem~\ref{theorem:MainTheoremDynamicProgramming},
we obtain the following corollary, improving Corollary~\ref{corollary:DVC}.
\begin{corollary}\label{corollary:DVC_general}
  {\sc Diverse Vertex Cover} can be solved on an input $(G,k,r,d)$ in time \[\Ocal(d\cdot |V(G)| \cdot (2^{k+2} \cdot (k+1))^r + |V(G)| \cdot  2^{k+1} \cdot (k+1)\cdot k).\]
\end{corollary}
Note that we obtain a slightly better running time than for Corollary~\ref{corollary:DVC}.
This is due to the fact that verifying the properties $E(G[X_t \setminus S]) = \emptyset$ and $\forall i \in \intv{1,{\delta(t)}}, S^i \cap X_t = S \cap X_{t_i}$ is done when constructing $\acore_{\rm VC}$ and not when constructing $\acore$.
Note also that, formally, we need to construct $\acore_{\rm VC}$ $r$ times but as it is $r$ times the same, we do the operation only once.

\section{Diversity in Kernelization}
Another key concept in the field of parameterized complexity is that of a \emph{kernelization} algorithm~\cite{kernelization}.
We have obtained some parallel results about the kernelization complexity of diverse problems as well
that we want to briefly sketch in this section.
A polynomial kernel of a parameterized problem is a polynomial-time algorithm
that given any instance either solves it or constructs in polynomial time an equivalent\footnote{
Meaning that the constructed instance is a \yes-instance if and only if the original instance was.}
instance whose size is polynomial in the parameter.
It is known that a parameterized problem is FPT if and only if it has a (not necessarily polynomial)
kernel, and a natural step after proving a parameterized problem to be FPT is to decide whether or not
it has a polynomial kernel.

We show that the diverse variants of several basic problems parameterized 
by the number of requested solutions plus solution size admit polynomial kernels as well.
This is done via a variant of the recently introduced notion of \emph{loss-less kernels}~\cite{CH16}
which are a special class of kernelizations that - very roughly speaking - 
for each but polynomially many bits of the input can either decide whether it 
has to be part of every solution or if it may be added to a solution without `destroying' it.

For instance, consider the famous Buss kernel for \textsc{Vertex Cover}~\cite{BG93}: 
Given a graph $G$ and an integer $k$, we want to decide if $G$ has a vertex cover of size $k$.
Each vertex of degree at least $k + 1$ must be in each solution. 
Otherwise, we have to include its (at least) $k+1$ neighbors, exceeding the size constraint.
On the other hand, each isolated (degree-$0$) vertex can be included in a vertex cover without
destroying it, but it does not cover any edge. 
In the `non-diverse' variant, we may remove these isolated vertices, 
and in the diverse variant, we have to keep some of them as
they may be used to increase the diversity.
However, polynomially (in $k$ and $r$) many such vertices suffice.

We now turn to the technical description of this framework.
All problems that fall into our framework have to be \emph{subset minimization problems}.
In a \emph{subset minimization problem}, one part of the input is a set, called the \emph{domain} of the instance, and the objective is to find a minimum size subset of the domain that satisfies a certain property.
For a subset minimization problem $\Pi$, and an instance $I$ of $\Pi$, we denote by $\domain(I)$ the \emph{domain} of $I$. E.g., in the \textsc{Vertex Cover} problem, an instance consists of a graph $G$ and an integer $k$ and the domain of the instance is $V(G)$. For an instance $(I, k)$ of a parameterized problem, we denote its domain by $\domain(I)$.

The following definition is a technical requirement to adapt loss-less kernelization to the setting of diverse problems. Domain recovery algorithms will be used to reintroduce some elements of the domain that have been removed during the kernelization process, in a controlled manner.

\begin{definition} Let $\Pi$ be a subset minimization problem. A \emph{domain recovery algorithm} takes as input two instances of $\Pi$, $I$ and $I'$, with $\domain(I') \subseteq \domain(I)$, and a set $S \subseteq \domain(I) \setminus \domain(I')$ and outputs in polynomial time an instance $\domrecovery_I(I', S)$
on domain $\domain(I') \cup S$, such that 
$\card{\domrecovery_I(I', S)} \le \card{I'} + g(\card{S})$  
for some computable function $g$. %
\end{definition}
We give the definition of a loss-less kernel~\cite{CH16}, tailored to our purposes as follows.\footnote{Due to technical reasons and at a potential cost of slightly increased kernel sizes, we do not keep track of the \emph{restricted} items that are forbidden in any solution of size $k$.} We use the following notation: For an instance $I$ of a subset minimization problem and an integer $k$, we denote by $\solutions(I, k)$ the solutions of $I$ of size at most $k$.
\begin{definition}\label{def:loss-less:kernel}
	Let $\Pi$ be a parameterized subset minimization problem. A \emph{loss-less kernelization} of $\Pi$ is a pair of a domain recovery algorithm and an algorithm that takes as input an instance $(I, k) \in \Sigma^* \times \bN$ and either correctly concludes that $(I, k)$ is a \no{}-instance, or outputs a tuple $(I', F, A)$ with the following properties. 
	$(I', k - \card{F})$ is an equivalent\footnote{I.e.\ $(I, k)$ is a \yes{}-instance if and only if $(I', k - \card{F})$ is.} instance to $(I, k)$ and
	$(F, A)$ is a partition of $\domain(I) \setminus \domain(I')$, and the following hold.
	\begin{enumerate}[(i)]
		\item\label{def:loss-less:kernel:size} There is a computable function $f$ such that $\card{I'} \le f(k)$.
		\item\label{def:loss-less:kernel:sol:reduce} For all $k' \le k$, for all $X \subseteq \domain(I)$, the following holds. Let $A'' \defeq X \cap A$. Then,	
		\begin{align*}
			&X \in \solutions(I, k') \Leftrightarrow F \subseteq X \mbox{ and } \\
			&X \setminus (F \cup A'') \in \solutions(I', k' - \card{F \cup A''}).
		\end{align*}
		\item\label{def:loss-less:kernel:sol:recovery} For all $k' \le k - \card{F}$, for all $X' \subseteq \domain(I')$, and for all $A'' \subseteq A' \subseteq A$ we have that:
      \[X' \in \solutions(I', k') \Leftrightarrow X \cup A'' \in
	  \solutions(\domrecovery_I(I', A'), k' + \card{A''}).\] 
	\end{enumerate}
  We call $f(k)$ the \emph{size} and $g(\cdot)$\footnote{Function $g$ is given implicitly in the domain recovery algorithm $\domrecovery_I(I', A')$.} the \emph{recovery cost} of the loss-less kernel, $F$ the \emph{forced} items and $A$ the \emph{allowed} items.
\end{definition}

We show that as a direct consequence of this definition, all elements in $A$ can be added to any solution to $(I, k)$ such that the resulting set remains a valid solution to $(I, k)$.

\begin{theorem}\label{thm:loss-less:diverse:kernel}
	Let $\Pi$ be a parameterized subset minimization problem that admits a loss-less kernel of size $f(k)$ and recovery cost $g(\cdot)$. Then, \textsc{Diverse $\Pi$} admits a kernel of size at most $f(k) + g(kr)$.
\end{theorem}
\begin{proof}
	Let $(I, k, r, \divT)$ be an instance of \textsc{Diverse $\Pi$}. Our algorithm works as follows. 
	We apply the loss-less kernel to $(I, k)$ and obtain $(I', F, A)$. Let $k' \defeq k - \card{F}$. 
	Then, we simply return 
	$(\domrecovery_I(I', A^*), k', r, \divT)$ 
	where $A^* = A$ if $\card{A} \le kr$ and otherwise, $A^*$ is an arbitrary size-$kr$ subset of $A$. 
	We now show that 
	$(\domrecovery_I(I', A^*), k', r, \divT)$ 
	is indeed an instance of \textsc{Diverse $\Pi$} 
	that is equivalent to $(I, k, r, \divT)$.
	
	Suppose $(I, k, r, \divT)$ is a \yes{}-instance. Then, there is a tuple $\cS = (S_1, \ldots, S_r) \in \domain(I)^r$ such that for all $i \in [1,r]$, $S_i \in \solutions(I, k)$ and $\diversity(\cS) \ge \divT$.
	
	\medskip
	\noindent\textbf{Case 1 ($\card{A} \le kr$).} In this case, $A^* = A$. For all $i \in [1,r]$, let $A_i \defeq S_i \cap A$, $S_i' \defeq S_i \setminus A_i$ and $S_i^* \defeq S_i' \setminus F$. By Definition~\ref{def:loss-less:kernel}(\ref{def:loss-less:kernel:sol:reduce}), we have that $S_i^* \in \solutions(I', k - \card{F} - \card{A_i})$. By Definition~\ref{def:loss-less:kernel}(\ref{def:loss-less:kernel:sol:recovery}), this implies that 
	$S_i' \in \solutions(\domrecovery_I(I', A^*), k')$ 
	(recall that $k' = k - \card{F}$).
	Furthermore, since $F \subseteq S_i$ for all $i \in [1,r]$ by Definition~\ref{def:loss-less:kernel}(\ref{def:loss-less:kernel:sol:reduce}), we have that $\diversity(S_1',\ldots, S_r') = \diversity(\cS) \ge \divT$, and hence 
	$(\domrecovery_I(I', A^*), k', r, \divT)$ 
	is a \yes{}-instance in this case.
	
	\medskip
	\noindent\textbf{Case 2 ($\card{A} > kr$).} In this case, $A^*$ is an arbitrary size-$kr$ subset of $A$. For all $i \in [1,r]$, let $A_i \defeq S_i \cap A$, $S_i^* \defeq S_i \setminus (F \cup A_i)$. 
	By Definition~\ref{def:loss-less:kernel}(\ref{def:loss-less:kernel:sol:reduce}) we have that $S_i^* \in \solutions(I', k' - \card{A_i})$.	
	Furthermore, since removing an element from some $S_i$ can decrease the diversity of the resulting solution by at most $(r-1)$, and since $F \subseteq S_i$ for all $i \in [1,r]$ by Definition~\ref{def:loss-less:kernel}(\ref{def:loss-less:kernel:sol:reduce}), we have that
	\begin{align*}
		\diversity(S_1^*, \ldots, S_r^*) \ge \diversity(\cS) - (r-1)\sum\nolimits_{i = 1}^r \card{A_i}.
	\end{align*}
	We construct a tuple of solutions to
	$\domrecovery_I(I', A^*)$
	as follows. Let $(B_1, \ldots, B_r)$ a tuple of pairwise disjoint subsets of $A^*$ such that for all $i \in [1,r]$, $\card{B_i} = \card{A_i}$. Such a tuple exists since $\sum_{i = 1}^r \card{A_i} \le kr = \card{A^*}$. For $i \in [1,r]$, let $S_i' \defeq S_i^* \cup B_i$ and $\cS' \defeq (S_1', \ldots, S_r')$. Let $i \in [1,r]$. Since $S_i^* \in \solutions(I', k' - \card{A_i})$, $\card{A_i} = \card{B_i}$ and $B_i \subseteq A^*$, we use 
	Definition~\ref{def:loss-less:kernel}(\ref{def:loss-less:kernel:sol:recovery}) to conclude that 
	$S_i' \in \solutions(\domrecovery_I(I', A^*), k')$.
	
	Now, adding $B_i$ to $S_i^*$ increased the diversity of the resulting solution by $(r-1)\cdot\card{A_i}$, since no element of $B_i$ is added to any other solution. 
	Hence,
	\begin{align*}
		\diversity(\cS') &= \diversity(S_1^*, \ldots, S_r^*) + (r-1)\sum\nolimits_{i = 1}^r \card{A_i} \\
						&\ge \diversity(\cS) \ge \divT.
	\end{align*}
	We have shown that 
	$(\domrecovery_I(I', A^*), k', r, \divT)$ 
	is a \yes{}-instance in this case as well.
	
	For the other direction, suppose 
	$(\domrecovery_I(I', A^*), k', r, \divT)$ 
	is a \yes{}-instance. Then, (\ref{def:loss-less:kernel:sol:reduce})  and (\ref{def:loss-less:kernel:sol:recovery}) of Definition~\ref{def:loss-less:kernel} immediately imply that $(I, k, r, \divT)$ is a \yes{}-instance as well.
	
	To bound the size of 
	$\domrecovery_I(I', A^*)$,  
	we have that $\card{I'} \le f(k)$ by the definition of the (loss-less) kernel, and 
	$\card{\domrecovery_I(I', A^*)} \le \card{I'} + g(\card{A^*}) \le f(k) + g(kr)$ 
	by the definition of a domain recovery algorithm.
\end{proof}
We now exemplify the use of
Theorem~\ref{thm:loss-less:diverse:kernel} by showing that several
well-known kernels hold in the diverse setting as well, giving
polynomial kernels in the parameterization solution size plus the
number of requested solutions.

We briefly introduce these problems.
In the \textsc{$d$-Hitting Set} problem, we are given a hypergraph $H$, each of whose hyperedges contains at most $d$ elements, and an integer $k$, and the goal is to find a set $S \subseteq V(H)$ of vertices of $H$ of size at most $k$ such that each hyperedge contains at least one element from $S$. In the \textsc{Point Line Cover} problem, we are given a set of points in the plane and an integer $k$, and we want to find a set of at most $k$ lines such that each point lies on at least one of the lines. A directed graph $D$ is called a \emph{tournament}, if for each pair of vertices $u, v \in V(D)$, either the edge directed from $u$ to $v$ or the edge directed from $v$ to $u$ is contained in the set of arcs of $D$. In the \textsc{Feedback Arc Set in Tournaments} problem we are given a tournament and an integer $k$, and the goal is to find a set of at most $k$ arcs such that after removing this set, the resulting directed graph does not contain any directed cycles.

\begin{theorem}\label{cor:diverse:kernels}
	The following diverse subset minimization problems parameterized by $k + r$ admit polynomial kernels.
	\begin{enumerate}[(i)]
		\item\label{cor:diverse:kernels:vc} \textsc{Diverse Vertex Cover}, on $\cO(k(k+r))$ vertices.
		\item\label{cor:diverse:kernels:dhs} \textsc{Diverse $d$-Hitting Set} for fixed $d$, on $\cO(k^d +kr)$ vertices.
		\item\label{cor:diverse:kernels:plc} \textsc{Diverse Point Line Cover}, on $\cO(k(k+r))$ points.
		\item\label{cor:diverse:kernels:fast} \textsc{Diverse Feedback Arc Set in Tournaments}, on $\cO(k(k+r))$ vertices.
	\end{enumerate}
\end{theorem}
\begin{proof}
(\ref{cor:diverse:kernels:vc})\footnote{This was also observed in~\cite{CH16}.} The classical kernelization for \textsc{Vertex Cover} due to~\cite{BG93} consists of the following two reduction rules. Let $(G, k)$ be an instance of \textsc{Vertex Cover}. First, we remove isolated vertices from $G$; since they do not cover any edges of the graph, we do not need them to construct a vertex cover. To obtain the loss-less kernel, we put these vertices into the set $A$. Second, if there is a vertex of degree more than $k$, this vertex has to be included in any solution; otherwise we would have to include its more than $k$ neighbors, resulting in a vertex cover that exceeds the size bound. We add this vertex to $F$, remove it from $G$ and decrease the parameter value by $1$. This second reduction rule finishes the description of the kernel. It is not difficult to argue that after an exhaustive application of these two rules, the resulting kernelized instance $(G', k')$ is such that either $k' < 0$, in which case we are dealing with a \no{}-instance, or $\card{V(G')} = \cO(k^2)$.
	For the domain recovery algorithm, we can use a trivial algorithm that reintroduces some of the vertices in $A$ to the graph $G'$.
	
	We now argue that this is indeed a loss-less kernel. Consider Definition~\ref{def:loss-less:kernel}. Item (\ref{def:loss-less:kernel:sol:reduce}) follows immediately from the fact that each vertex cover of $G$ of size at most $k$ has to contain all vertices in $F$ and that each vertex in $A$ has no neighbors in $V(G')$. The latter also implies (\ref{def:loss-less:kernel:sol:recovery}).
	The result now follows from Theorem~\ref{thm:loss-less:diverse:kernel}.
	
	(\ref{cor:diverse:kernels:dhs}) We show that the kernel on $\cO(k^d)$ vertices presented in~\cite[Section 2.6.1]{Cygan} is a loss-less kernel. This kernel is essentially a generalization of the one presented in the proof of (\ref{cor:diverse:kernels:vc}), so we will skip some of the details.
	It is based on the following reduction rule: If there are $k+1$ hyperedges $e_1, \ldots, e_{k+1}$ with $Y \defeq \bigcap_{i = 1}^{k+1} e_i$ such that for each $i \in [k+1]$, $e_i \setminus Y \neq \emptyset$, then any solution has to contain $Y$; otherwise, to hit the hyperedges $e_1, \ldots, e_{k+1}$, we would have to include at least $k+1$ elements in the hitting set. Moreover, if $Y = \emptyset$, we can immediately conclude that we are dealing with a \no{}-instance. If $Y$ is nonempty, then we add all elements of $Y$ to $F$ and decrease the parameter value by $\card{Y}$. The set $A$ consists of all vertices that are isolated (i.e.\ not contained in any hyperedge) after exhaustively applying the previous reduction rule. Following the same argumentation above (and using the same domain recovery algorithm), we can conclude that this procedure is a loss-less kernel on $\cO(k^d)$ vertices, and the result follows from Theorem~\ref{thm:loss-less:diverse:kernel}.
	
	(\ref{cor:diverse:kernels:plc}) Let $(P, k)$ be an instance of \textsc{Point Line Cover}. We consider the set of the lines defined by all pairs of points of $P$ as the domain of $(P, k)$, and we denote this set by $L(P)$. All solutions to $(P, k)$ can be considered a subset of $L(P)$. We obtain a kernel on $\cO(k^2)$ points as follows, cf.~\cite[Exercise 2.4]{Cygan}. The idea is again similar to the kernel presented in~(\ref{cor:diverse:kernels:vc}). If there are $k+1$ points on a line, then we have to include this line in any solution; we add such lines to the set $F$ and remove all points on them from $P$, and decrease the parameter value by $1$. We finally add to $A$ all lines that have no points on them. We can argue in the same way as above that this gives a kernel with at most $\cO(k^2)$ points and with Theorem~\ref{thm:loss-less:diverse:kernel}, the result follows.
	
	(\ref{cor:diverse:kernels:fast}) We observe that the kernel given in~\cite[Section 2.2.2]{Cygan} is a loss-less kernel. Its first reduction rule states that if there is an arc that is contained in at least $k+1$ triangles, then we reverse this arc and decrease the parameter value by $1$, and the second reduction rule states that any vertex that is not contained in a triangle can be removed. Any arc affected by the former rule will be put in the set $F$ and any arc affected by the latter rule will be put in the set $A$.
	We now describe the domain recovery algorithm. Let $(T, k)$ be the original instance and $(T', k')$ the kernelized instance, and let $(u, v) = a \in A$ be an arc. Then, we add $a$ to $T'$ and to ensure that the resulting directed graph is a tournament, for any $x \in \{u, v\} \setminus V(T')$, we add all arcs $(x, y) \in E(T)$ and $(y, x) \in E(T)$ to $T'$. Since $a \in A$, we know that one of its endpoints was not contained in any triangle, and hence adding the endpoints of $a$ and all their incident arcs does not add any triangles to the tournament. 
\end{proof}

We would like to remark that the crucial part to use loss-less kernels in the diverse setting was that \emph{any} solution of size at most $k$ has to contain all vertices of $F$, and arbitrarily adding vertices from $A$ does not destroy a solution. In the `classical' kernelization setting, 
to argue that a reduction rule is safe it is sufficient to show that the existence of a vertex cover in the original instance implies the existence of \emph{some} vertex cover in the reduced instance and vice versa, see e.g.,~\cite{FJKRW18,kernelization}. This alone is usually not enough to argue that a reduction preserves diverse solutions.

\section{Conclusion}

In this work, we considered a formal notion of diversity of a set of solutions 
to combinatorial problems in the setting of parameterized algorithms.
We showed how to emulate treewidth based dynamic programming algorithms in order to solve diverse problems in FPT time, 
with the number $r$ of requested solutions being an additional parameter. %

This line of research is now wide open, with many natural questions to address. As all our results are of a positive nature, we ask: when can diversity be a source of hardness?
Concretely, a natural target in parameterized complexity would be to identify a parameterized problem $\Pi$ that is FPT, however \textsc{Diverse $\Pi$} being W[1]-hard when $r$ is an additional parameter. For positive results, an interesting research direction would be to generalize our framework for diverse problems to other well-studied width measures for graphs, as well as to other structures, such as matroids.

In this work, we considered the \emph{sum} of all pairwise Hamming distances of a set as a measure of diversity.
As pointed out, this measure has some weaknesses, and another widely used measure is the \emph{minimum} Hamming distance.
In this setting, we only obtain FPT-results when the diversity is bounded by 
a function of the treewidth and the number of solutions, but not in general.
So, a natural follow-up question is whether or not we can obtain FPT-results 
under the minimum Hamming distance, even if the diversity is unbounded.

\subparagraph*{Acknowledgements.}
M.\ Fellows and F.\ Rosamond acknowledge support from the Norwegian NFR
Toppforsk Project, ``Parameterized Complexity and Practical Computing''
(PCPC) (no.\ 813299).
M.\ Fellows, F.\ Rosamond, L.\ Jaffke, M.\ de Oliveira Oliveira, and G. Philip
acknowledge support from the Bergen Research Foundation grant ``Putting
Algorithms Into Practice'' (BFS, no.\ 810564).
G.\ Philip acknowledges support from the Research Council of Norway grants
``Parameterized Complexity for Practical Computing'' (NFR, no.\ 274526d),
``MULTIVAL'', and ``CLASSIS'', and European Research Council (ERC) under the
European Union's Horizon 2020 research and innovation programme (ERC, no.\
819416).
T.\ Masa\v{r}\'{i}k acknowledges support from the European Research Council
(ERC grant agreement no.\ 714704)
~\includegraphics[height=15px]{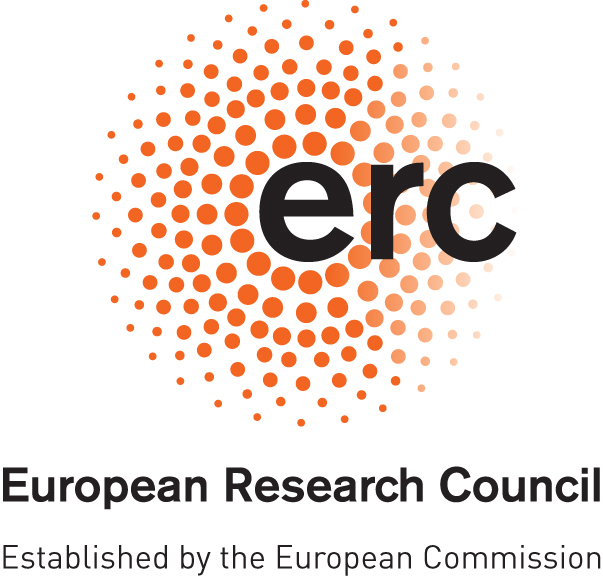}%
~\includegraphics[height=15px]{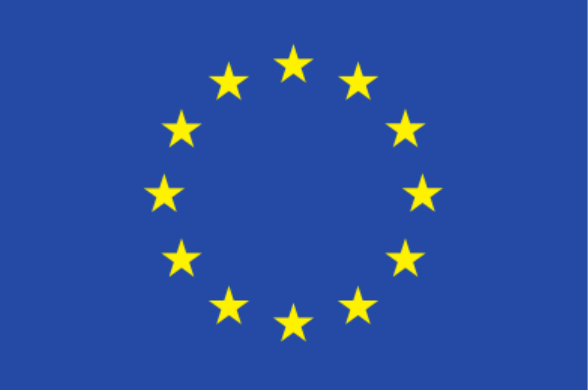}%
,
and from Charles University’s grants GAUK 1514217 and SVV-2017-260452.
He recently started a postdoc at Simon Fraser University, Canada.
M.\ de Oliveira Oliveira acknowledges support from the Research Council of
Norway (NFR, no.\ 288761).
J.\ Baste acknowledges support from the  German Research Foundation
(DFG, no.\ 388217545).

\bibliographystyle{plainurl}
\bibliography{diverse-full}

\end{document}